\documentclass[10pt,conference,onecolumn]{IEEEtran}
\IEEEoverridecommandlockouts
\usepackage[final]{graphicx}
\usepackage{cite}
\usepackage{amssymb}
\usepackage{amsbsy}
\usepackage{amsmath}
\usepackage{amsthm}
\usepackage{amsfonts}
\usepackage{bm}
\usepackage{bbm}
\usepackage{dsfont}
\usepackage[linesnumbered,ruled,vlined]{algorithm2e}
\usepackage{setspace}

\begin{document}
\setstretch{1.5}
\title{\LARGE{A Worst-Case Performance Optimization Based Design Approach to Robust Symbol-Level Precoding for Downlink MU-MIMO}\vspace{-.3cm}}

\newcommand\Mark[1]{\textsuperscript#1}

\author{
        \IEEEauthorblockN{Alireza~Haqiqatnejad\Mark{1}, Shahram~Shahbazpanahi\Mark{1}\Mark{,}\Mark{2}, and Bj\"{o}rn~Ottersten\Mark{1}}
        \IEEEauthorblockA{\Mark{1}Interdisciplinary Centre for Security, Reliability and Trust (SnT), University of Luxembourg
                \\ Email: \texttt{\{alireza.haqiqatnejad,bjorn.ottersten\}@uni.lu}}
        \IEEEauthorblockA{\Mark{2}Department of Electrical, Computer, and Software Engineering, University of Ontario Institute of Technology
                \\Email: \texttt{shahram.shahbazpanahi@uoit.ca}}
        \thanks{\fontsize{7.5}{8.5}\selectfont{The authors are supported by the Luxembourg National Research Fund (FNR) under CORE Junior project: C16/IS/11332341 Enhanced Signal Space opTImization for satellite comMunication Systems (ESSTIMS).}}
        \vspace{-.7cm}
}

\newtheorem{theorem}{Theorem}
\newtheorem{acknowledgement}[theorem]{Acknowledgement}
\newtheorem{axiom}[theorem]{Axiom}
\newtheorem{case}[theorem]{Case}
\newtheorem{claim}[theorem]{Claim}
\newtheorem{conclusion}[theorem]{Conclusion}
\newtheorem{condition}[theorem]{Condition}
\newtheorem{conjecture}[theorem]{Conjecture}
\newtheorem{corollary}[theorem]{Corollary}
\newtheorem{criterion}[theorem]{Criterion}
\newtheorem{definition}[theorem]{Definition}
\newtheorem{example}[theorem]{Example}
\newtheorem{exercise}[theorem]{Exercise}
\newtheorem{lemma}[theorem]{Lemma}
\newtheorem{notation}[theorem]{Notation}
\newtheorem{problem}[theorem]{Problem}
\newtheorem{proposition}[theorem]{Proposition}
\newtheorem{remark}[theorem]{Remark}
\newtheorem{solution}[theorem]{Solution}
\newtheorem{summary}[theorem]{Summary}
\newtheorem{property}{Property}

\def\bh{{\bf h}}
\def\bu{{\bf u}}      
\def\bt{{\bf t }}
\def\bw{\mathrm{\mathbf{w}}}

\newcommand{\defeq} {\overset{\underset{\mathrm{def}}{}}{=}}
\newcommand{\Deee} {\mathrm{\boldsymbol{\delta}}}
\newcommand{\lamb} {\mathrm{\mathbf{\Lambda}}}
\newcommand{\psiii} {\mathrm{\boldsymbol{\psi}}}
\newcommand{\ups} {\mathrm{\mathbf{\upsilon}}}
\newcommand{\Tht} {\mathrm{\mathbf{\Theta}}}
\newcommand{\Piii} {\mathrm{\mathbf{\Pi}}}
\newcommand{\Theee} {\mathrm{\mathbf{\theta}}}
\newcommand{\Lamb} {\mathrm{\mathbf{\Lambda}}}
\newcommand{\g} {\mathrm{\mathbf{\gamma}}}
\newcommand{\Gammaaa}{\mathrm{\mathbf{\Gamma}}}
\newcommand{\Lam}{\mathrm{\mathbf{\Lambda}}}
\newcommand{\Sigmaaa}{\mathrm{\mathbf{\Sigma}}}
\newcommand{\Siii}{\mathrm{\mathbf{\psi}}}
\newcommand{\Varrr}{\mathrm{\pmb{\vartheta}}}
\newcommand{\Phiii}{\mathrm{\mathbf{\Phi}}}
\newcommand{\chiii}{\mathrm{\mathbf{\chi}}}
\newcommand{\Mu} {\mathrm{\mathbf{\mu}}}
\newcommand{\omg} {\mathrm{\mathbf{\omega}}}
\newcommand{\HHH} {\mathrm{\mathbf{H}}}
\newcommand{\QQQ} {\mathrm{\mathbf{Q}}}
\newcommand{\bbb}{\mathrm{\mathbf{b}}}
\newcommand{\uuu}{\mathrm{\mathbf{u}}}
\newcommand{\ddd}{\mathrm{\mathbf{d}}}
\newcommand{\gggg}{\mathrm{\mathbf{g}}}
\newcommand{\EEE}{\mathrm{\mathbf{E}}}
\newcommand{\WWW}{\mathrm{\mathbf{W}}}
\newcommand{\BBB}{\mathrm{\mathbf{B}}}
\newcommand{\ZZZ}{\mathrm{\mathbf{Z}}}
\newcommand{\I}{\mathrm{\mathbf{I}}}
\newcommand{\J}{\mathrm{\mathbf{J}}}
\newcommand{\A}{\mathrm{\mathbf{A}}}
\newcommand{\DDD}{\mathrm{\mathbf{D}}}
\newcommand{\FFF}{\mathrm{\mathbf{F}}}
\newcommand{\G}{\mathrm{\mathbf{G}}}
\newcommand{\T}{{T}}
\newcommand{\F}{\mathrm{F}}
\newcommand{\HH}{\mathrm{H}}
\newcommand{\REAL}{\mathrm{Re}}
\newcommand{\IMAG}{\mathrm{Im}}
\newcommand{\LLL}{\mathrm{\mathbf{L}}}
\newcommand{\R}{\mathrm{\mathbf{R}}}
\newcommand{\YYY}{\mathrm{\mathbf{Y}}}
\newcommand{\CCC}{\mathrm{\mathbf{C}}}
\newcommand{\XXX}{\mathrm{\mathbf{X}}}
\newcommand{\MMM}{\mathrm{\mathbf{M}}}
\newcommand{\PPP}{\mathrm{\mathbf{P}}}
\newcommand{\RRR}{\mathrm{\mathbf{R}}}
\newcommand{\GGG}{\mathrm{\mathbf{G}}}
\newcommand{\TTT}{\mathrm{\mathbf{T}}}
\newcommand{\OOO}{\mathrm{\mathbf{0}}}
\newcommand{\UUU}{\mathrm{\mathbf{U}}}
\newcommand{\aaa}{\mathrm{\mathbf{a}}}
\newcommand{\h}{\mathrm{\mathbf{h}}}
\newcommand{\qqq}{\mathrm{\mathbf{q}}}
\newcommand{\eee}{\mathrm{\mathbf{e}}}
\newcommand{\s}{\mathrm{\mathbf{s}}}
\newcommand{\ggggg}{\mathrm{\mathbf{g}}}
\newcommand{\vvv}{\mathrm{\mathbf{v}}}
\newcommand{\ttt}{\mathrm{\mathbf{t}}}
\newcommand{\fff}{\mathrm{\mathbf{f}}}
\newcommand{\zzz}{\mathrm{\mathbf{z}}}
\newcommand{\ccc}{\mathrm{\mathbf{c}}}
\newcommand{\x}{\mathrm{\mathbf{x}}}
\newcommand{\yyy}{\mathrm{\mathbf{y}}}
\newcommand{\ppp}{\mathrm{\mathbf{p}}}
\newcommand{\www}{\mathrm{\mathbf{w}}}
\newcommand{\mmm}{\mathrm{\mathbf{m}}}
\newcommand{\rrr}{\mathrm{\mathbf{r}}}
\newcommand{\CI}{\scriptscriptstyle{(\mathrm{CI}})}
\newcommand{\ML}{\scriptscriptstyle{(\mathrm{ML}})}
\newcommand{\EXP}{\mathds{E}}
\newcommand{\PR}{\mathrm{Pr}}
\newcommand{\TR}{\mathrm{Tr}}
\newcommand{\rank}{\mathrm{rank}}
\newcommand{\VEC}{\mathrm{vec}}
\newcommand{\SUP}{\mathrm{sup}}
\newcommand{\INF}{\mathrm{inf}}
\newcommand{\DET}{\mathrm{det}}
\newcommand{\TT}{\mathrm{T}}
\newcommand{\C}{\mathds{C}}
\newcommand{\Q}{\mathds{Q}}
\newcommand{\RR}{\mathds{R}}
\newcommand{\D}{\mathcal{D}}
\newcommand{\SPSK}{\mathrm{-S}\mathrm{PSK}}
\newcommand{\SQPSK}{\mathrm{-S}\mathrm{QPSK}}
\newcommand{\SAPSK}{\mathrm{-S}\mathrm{APSK}}
\newcommand{\SkPSK}{\mathrm{-S}^k\mathrm{PSK}}
\newcommand{\SkQPSK}{\mathrm{-S}^k\mathrm{QPSK}}
\newcommand{\SkAPSK}{\mathrm{-S}^k\mathrm{APSK}}
\newcommand{\PSNR}{\mathrm{PSNR}}
\newcommand{\SNR}{\mathrm{SNR}}
\newcommand{\LLR}{\mathrm{LLR}}
\newcommand{\diag}{\mathop{\mathrm{diag}}}

\newlength{\dhatheight}
\newcommand{\doublehat}[1]{%
        \settoheight{\dhatheight}{\ensuremath{\hat{#1}}}%
        \addtolength{\dhatheight}{-0.3ex}%
        \hat{\vphantom{\rule{1pt}{\dhatheight}}%
                \smash{\hat{#1}}}}

\newcommand*{\Scale}[2][4]{\scalebox{#1}{$#2$}}%
\SetKwInput{KwInput}{input}
\SetKwInput{KwOutput}{output}
\SetKwInput{KwInitialize}{initialize}
\SetKwInput{KwSet}{set}
\SetKwRepeat{Do}{do}{until}
\newcommand{\nonl}{\renewcommand{\nl}{\let\nl\oldnl}}

\maketitle

\begin{abstract}
This paper addresses the optimization problem of symbol-level precoding (SLP) in the downlink of a multiuser multiple-input multiple-output (MU-MIMO) wireless system while the precoder's output is subject to partially-known distortions. In particular, we assume a linear distortion model with bounded additive noise. The original signal-to-interference-plus-noise ratio (SINR) -constrained SLP problem minimizing the total transmit power is first reformulated as a penalized unconstrained problem, which is referred to as the relaxed robust formulation. We then adopt a \textit{worst-case design approach} to protect the users' intended symbols and the targeted constructive interference with a desired level of confidence. Due to the non-convexity of the relaxed robust formulation, we propose an iterative algorithm based on the block coordinate ascent-descent method. We show through  simulation results that the proposed robust design is flexible in the sense that the CI constraints can be relaxed so as to keep a desirable balance between achievable rate and power consumption. Remarkably, the new formulation yields more energy-efficient solutions for appropriate choices of the penalty parameter, compared to the original problem.
\end{abstract}

\vspace{-.1cm}

\section{Introduction}

Multiuser precoding is well known to be an effective way of handling multiuser interference which is a limiting factor while simultaneously serving multiple user equipments in  the same time/frequency resource block. Beyond the wide variety of block-level precoding techniques  proposed in the literature (see e.g., \cite{tb_opt,vec_per,tb_convex,tb_sinr} and the references therein), processing the transmit signal in a symbol-by-symbol fashion can lead to improvements in spectral/energy efficiency, at the price  of increased transmitter complexity \cite{slp_chr,slp_con}. In this (non-linear) design approach, which is commonly referred to as symbol-level precoding (SLP), the precoded transmit signal is optimized with respect to the instantaneous channel as well as the instantaneous users' data symbols.

A key consideration in designing the symbol-level precoder is to properly define the constructive interference (CI) regions based on the received signal constellation, typically with the aim of preserving (or enhancing) the detection accuracy \cite{slp_chr,slp_con,slp_gen}. The precoding scheme then allows a (noise-free) received symbol to be observed anywhere within the correct CI region. This type of design, however, is highly sensitive to inaccuracies in several parameters, such as the available channel state information at the transmitter (CSIT), the receive noise power, and any succeeding operation on the transmit signal which is not perfectly known to the precoder. More specifically, considering (non)linear distortions of the precoded signal, which falls within the third category, is the main focus of this paper. The distorted transmit signal may reflect the effects of non-ideal elements either in the digital domain, e.g., low-resolution digital-to-analog converters (DAC), or in the RF chain, e.g., power amplifiers \cite{sp_over}. Furthermore, it can be an adequate modeling of the source-relay link over a relay channel, e.g., the feeder link in a satellite communication system.

There has been some recent effort addressing the SLP deign problem in the presence of system uncertainties. The earliest work in \cite{unc_chr} considers the noisy received signal and proposes a stochastic robust approach with probabilistic (rate) outage constraints. Robust symbol-level precoders under imperfect CSIT are presented in \cite{slp_chr} and \cite{slp_robust}. In \cite{slp_chr}, a worst-case robust SLP scheme is proposed under bounded CSIT errors for quality-of-service (QoS) -constrained power mininization  and max-min fairness design criteria. In addition to bounded channel uncertainty modeling, the SLP problem with statistically-known CSIT is addressed in \cite{slp_robust}, where deterministic convex approximations are derived for the probabilistic CI constraints. \textit{To the best of authors' knowledge, the SLP design problem under linear distortion of the precoded signal has not been addressed in the literature}. In this paper, by assuming a linearly distorted signal model with bounded additive distortion, we aim to design an SLP scheme such that the performance gain offered by the CI-based design is preserved. In particular, we reformulate a version of the original problem with penalized objective function and use this reformulation in a worst-case design approach. The penalty coefficient in the new formulation allows us to keep a balance between the desired level of spectral efficiency/users' symbol error probability and the consumed power. 

It is worth mentioning that the problem of robust design has been widely studied in the literature for scenarios where our knowledge about the environment is subject to uncertainty \cite{mehrzad2004rdp,Chalise2007,Shahbazpanahi2003,
Vorobyov2,Vorobyov4,ShahramMUD,Zarifi,YueShahramOSTBC}. In this paper, \textit{we assume that our design process is subject to uncertainty}, e.g., due to finite precision of the underlying design and implementation technology. This work can point the research community to address new practical challenges in robust design when the design parameters are subject to uncertainty.

The remainder of this paper is organized as follows. In Section \ref{sec:sys}, we describe the system and signal distortion models along with the original SLP problem formulation. In Section \ref{sec:wc}, we reformulate and discuss the worst-case design problem and present our proposed algorithm. Simulation results are presented in Section \ref{sec:sim}. Finally, Section \ref{sec:con} concludes the paper.

\noindent{\bf{Notations:}} We use uppercase and lowercase bold-faced letters to denote matrices and vectors, respectively. For a matrix $\A$, $\mathrm{rank}(\A)$ denotes the column rank of $\A$. For matrices and vectors, $\|\cdot\|$ respectively denotes the spectral norm and the Euclidean norm. Operators $\mathrm{diag}(\cdot)$ and $\mathrm{blkdiag}(\cdot)$ represent diagonal  and block-diagonal matrices, respectively. We use $\mathbf{I}$ and $\OOO$ to represent, respectively, the identity matrix and the zero matrix (or the zero vector, depending on the context) of appropriate dimensions. The operator $\otimes$ stands for the Kronecker product.
Statistical expectation is denoted as $\EXP\{\cdot\}$.

\section{System Model and Problem Definition}\label{sec:sys}

\begin{figure}
        \centering
        \includegraphics[trim={.5in 2in .5in 2in},clip,width=.8\columnwidth]{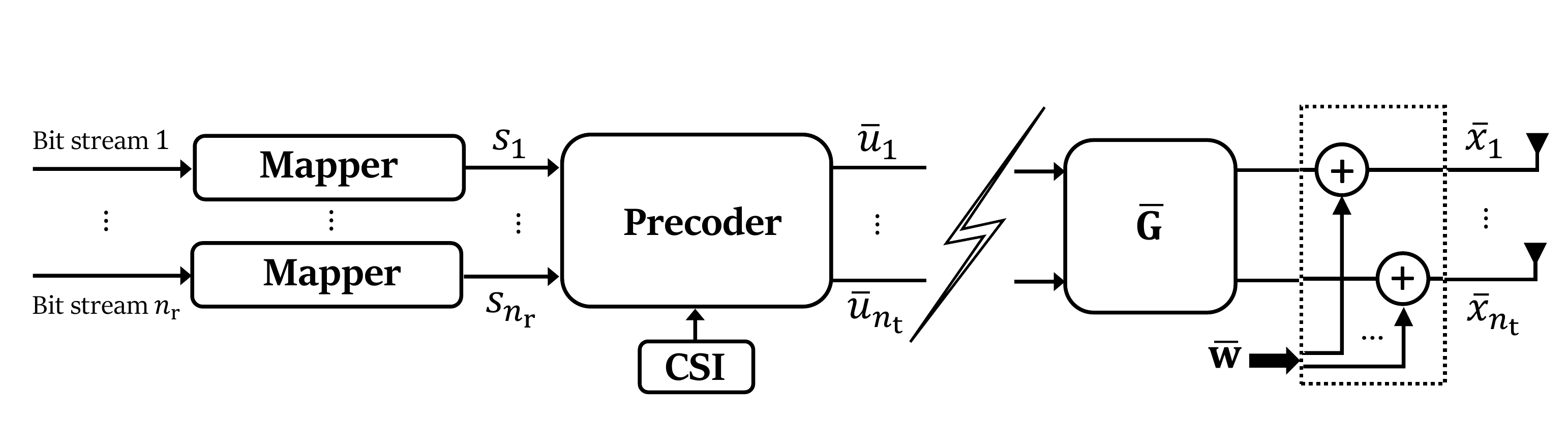}
        \caption{The considered system model: the output of the symbol-level precoder, $\bar{\uuu}$, is subject to linear distortion before being transmitted to the UEs, while the additive distortion $\bar{\www}$ is not perfectly known to the precoder.}
        \label{fig:sysmodel}
\end{figure}
We consider the downlink of a multiuser multiple-input multiple-output (MU-MIMO) system in which a common transmitter, equipped with an array of $n_\mathrm{t}$ antennas, communicates with $n_\mathrm{r}$ single-antenna user equipments (UE) through sending (independent) data streams. Number of simultaneously-served UEs is limited by $n_\mathrm{t}$, i.e., $n_\mathrm{r}\leq n_\mathrm{t}$. A quasi-static frequency-flat fading channel is considered between any transmitter/receiver antenna pair. The discrete-time data symbols $\{s_i\}_{i=1}^{n_\mathrm{r}}$ are assumed to be drawn from a finite-alphabet equiprobable constellation set with unit average power, where $s_i$ denotes the symbol intended for the $i$th UE. The transmitter employs a multiuser precoder to spatially multiplex the users' data streams in the downlink transmission. We adopt a symbol-level (non-linear) precoding scheme \cite{slp_chr,slp_con,slp_tsp}, in which the $n_\mathrm{t}\times1$ precoded signal $\bar{\uuu}=[\bar{u}_1,\ldots,\bar{u}_{n_\mathrm{t}}]^\T$ is redesigned every symbol period by solving an optimization problem. It is further assumed that the precoded signal is subject to linear distortion before transmission, i.e., the actual $n_\mathrm{t}\times1$ transmitted signal $\bar{\x}$ is
given by
\begin{equation}\label{eq:u}
\bar{\x} = \bar{\GGG} \bar{\uuu} + \bar{\www},
\end{equation}
where $\bar{\GGG}\in\mathds{C}^{n_\mathrm{t}\times n_\mathrm{t}}$ is a known distortion matrix and $\bar{\www}\in\mathds{C}^{n_\mathrm{t}\times 1}$ represents an additive white noise which is uncorrelated with the precoder's output $\bar{\uuu}$. Such a model is particularly suitable for a relayed transmission scheme. For example, interference mitigation techniques in the forward link of a satellite communication system may take the form of on-ground precoding, i.e., the UEs' data streams are pre-processed at the gateway and then sent to the satellite through the feeder link \cite{sat_hybrid}. The received signal by the satellite (to be transmitted towards the UEs) can be modeled as \eqref{eq:u}, where $\bar{\GGG}$  represents the atmospheric attenuation and $\bar{\www}$ models the additive noise at the satellite reflector antennas. Another possible application of \eqref{eq:u} could be in a massive MU-MIMO scenario where the (continuous-valued) precoding coefficients 
$\{\bar{u}_j\}_{j=1}^{n_{\rm t}}$ are passed through low-resolution DACs to be quantized in the digital domain before up-conversion via the RF chains. The non-linear quantization operation can then be approximated by the additive quantization noise (AQN) model, \cite{vec_quan,pre_quan}, which coincides with the linear distortion model in \eqref{eq:u}. Under the above assumptions, the baseband representation of the signal observed by the $i$th UE is given as
\begin{equation}\label{eq:sys}
r_i = \h_i^\T \bar{\x}+z_i = \h_i^\T (\bar{\GGG} \bar{\uuu} + \bar{\www})+z_i, \; i=1,...,n_\mathrm{r},
\end{equation}
where the vector $\h_i\in\C^{n_\mathrm{t}\times 1}$ contains the instantaneous fading coefficients of the channel between the transmit antennas and the $i$th UE, and $z_i\sim\mathcal{CN}(0,\sigma_i^2)$ models the additive thermal noise at the $i$th receive front-end. Note that, unlike conventional linear precoding techniques, the precoded signal $\bar{\uuu}$ may not explicitly be decomposed in terms of distinct users' signatures (i.e., precoding vectors); therefore, received SINR is no longer an applicable measure. Instead, we use the received SNRs $\left\{\bar{\x}^T\h_i\h_i^T\bar{\x}/\sigma_i^2\right\}_{i=1}^{n_\mathrm{r}}$ as the QoS measure. 

We define equivalent real-valued notations: $\!\uuu\!\triangleq\![\REAL(\bar{\uuu})^\T\!,\!\IMAG(\bar{\uuu})^\T]^\T\!$, $\x\!\triangleq\![\REAL(\bar{\x})^\T,\IMAG(\bar{\x})^\T]^\T$, $\www\!\triangleq\![\REAL(\bar{\www})^\T,\IMAG(\bar{\www})^\T]^\T$, and for all $i=1,...,n_\mathrm{r}$, we denote $\s_i\triangleq[\REAL(s_i),\IMAG(s_i)]^\T$, $\HHH_i\triangleq T(\h_i^\T)$, and $\GGG\triangleq[\GGG_1^\T,...,\GGG_{n_\mathrm{t}}^\T]^\T$ with $\GGG_j\triangleq T(\ggggg_j^\T)$ and $\ggggg_j$ denoting the $j$th column of $\bar{\GGG}^\T$ for $j=1,...,n_\mathrm{t}$, where
\begin{equation}\nonumber
\begin{aligned}
T(\vvv) &\triangleq \begin{bmatrix}
\REAL(\vvv) \; -\IMAG(\vvv)\\
\IMAG(\vvv) \quad\;\: \REAL(\vvv)
\end{bmatrix},
\end{aligned}
\end{equation}
for any complex input vector $\vvv$. Using these new notations, it is easy to verify that $\x=\GGG\uuu+\www$ holds true, and thus, the $i$th \textit{real-valued noise-free} received signal can be represented as a $2\times 1$ vector given by  $\HHH_i\x=\HHH_i(\GGG\uuu+\www)$.
 It is worth mentioning that  the additive distortion vector $\www$, without any restriction on its distribution, is assumed to be norm-bounded, i.e., $\|\www\|\leq \varepsilon$.  

To exploit the instantaneous data information in a symbol-level precoded transmission, the precoder is typically designed such that each noise-free received signal $\HHH_i\x,i=1,...,n_\mathrm{r}$, is observed within a pre-defined region corresponding to the intended symbol $\s_i$, called constructive interference (CI) region. The CI regions have been defined in several ways in the literature; see, e.g., \cite{slp_chr,slp_con,slp_gen}. We focus on the so-called distance-preserving CI regions \cite{slp_gen}, in which any point belonging to the CI region of $\s_i$  is closer to  $\s_i$ than any other constellation point. Such a definition implicitly assumes that the $i$th UE uses off-the-shelf optimal single-user detectors, e.g., single-user maximum likelihood (ML)\ receiver.

Given the set of target SNRs $\{\gamma_i\}_{i=1}^{n_\mathrm{r}}$ to be achieved for all UEs, a well-known design criterion is to minimize the instantaneous (per-symbol) total transmit power while satisfying the CI constraint for every single UE. Denoting the distance-preserving CI region associated with $\s_i$ by $\mathcal{D}_i(\s_i,\gamma_i,\sigma_i)$, the optimal transmit signal $\x$ is then obtained by solving the following optimization
problem:\begin{equation}\label{eq:slp0}
\underset{\x}{\min} \quad \x^\T\x \quad\mathrm{s.t.} \quad  \HHH_i \x \in \mathcal{D}(\s_i,\gamma_i,\sigma_i), \; i=1,...,n_\mathrm{r}.
\end{equation}
It has been shown in \cite{slp_tsp} that problem \eqref{eq:slp0} can be formulated, in a compact form, as a linearly-constrained quadratic program (QP):
\begin{equation}\label{eq:slp}
\underset{\x,\ttt\succeq\OOO}{\min} \quad \x^\T\x \quad \mathrm{s.t.} \quad \HHH \x = \DDD \s + \A^{-1}\WWW\ttt,
\end{equation}
where we use the following definitions: 
$\HHH \triangleq [\HHH_1^\T,...,\HHH_{n_\mathrm{r}}^\T]^\T$; $\A \triangleq \mathrm{blkdiag}(\A_1,...,\A_{n_\mathrm{r}})$; $\A_i \!=\! [\aaa_{i,1},\aaa_{i,2}]^\T\!\in\!\mathds{R}^{2\times2}$ contains the normal vectors of the ML decision boundaries; $\DDD \triangleq \diag(\sigma_1\sqrt{\gamma_1},...,\sigma_{n_\mathrm{r}}\sqrt{\gamma_{n_\mathrm{r}}})\otimes \I_2$, $\s\triangleq[\s_1,...,\s_{n_\mathrm{r}}]^\T$; and $\WWW$ is a diagonal binary weighting matrix with a diagonal element being one if the corresponding symbol is an outer constellation point and zero otherwise, and $\ttt\!\in\!\mathds{R}_+^{2n_\mathrm{r}\times1}$ is a slack variable with a geometrical interpretation behind. Indeed, the entries of $\ttt$ specify the (orthogonal) distances between the received symbols and their corresponding CI boundaries. The larger the elements of $\ttt$, the deeper the received symbol is pushed into the correct decision region. Note also that $\A$ can always be formed as a full-rank square, and  hence invertible, matrix.



\section{Worst-Case Design Formulation}\label{sec:wc}

We start off by casting a new optimization problem other than \eqref{eq:slp} by introducing the linear equality CI constraints as an $\ell_2$-norm penalty into the objective function, i.e.,
\begin{equation}\label{eq:slp2}
\begin{aligned}
\underset{\x,\ttt\succeq\OOO}{\min} \quad & \|\x\|^2+\beta\,\|\HHH\x - \DDD \s - \A^{-1}\ttt\|^2,
\end{aligned}
\end{equation}
where $\beta$ denotes the penalty coefficient. It is worth noting that unlike \eqref{eq:slp}, this new formulation does not strictly impose the CI constraints. Instead, the $\ell_2$-norm term in the objective function penalizes any feasible solution for which the received symbols will not exactly be located within the intended CI regions. For this reason, we refer to problem \eqref{eq:slp2} as the relaxed SLP design. Intuitively speaking, setting larger values for $\beta$ puts more emphasis on the satisfaction of CI constraints (i.e., more severely penalizes any deviation of the received symbols from the correct CI regions), but  may lead to higher transmission powers. This introduces a tradeoff in choosing the penalty parameter $\beta$, where its effect on the performance will be investigated  via simulation results in Section \ref{sec:sim}. It is also worth noting that problem \eqref{eq:slp2} becomes equivalent to \eqref{eq:slp} as $\beta\rightarrow\infty$.

By replacing $\x$ with $\GGG\uuu+\www$ in \eqref{eq:slp2}, we are ready to define the worst-case design formulation of our interest:
\begin{equation}\label{eq:main}
\begin{aligned}
\underset{\uuu,\ttt\succeq\OOO}{\min} \quad \underset{\|\www\|\leq \varepsilon}{\max} \quad & \|\GGG\uuu+\www\|^2 + \beta\, \|\HHH(\GGG\uuu+\www) - \Phiii(\ttt)\|^2,
\end{aligned}
\end{equation}
where $\Phiii(\ttt)\triangleq\DDD\s + \A^{-1}\ttt$. The optimization problem \eqref{eq:main} is non-convex, and thus, may not be amenable to a computationally efficient solution. To tackle the optimization problem \eqref{eq:main}, we propose a three-step iterative block coordinate ascent-descent algorithm: in the first step, the inner maximization is solved for given $\uuu$  and $\ttt \succeq\OOO  $, thereby obtaining a new value for $\www$ in a \textit{semi-closed} form in terms of $\bu$ and $\ttt$. In the second step, the value of $\bt$ is updated by solving a non-negative least squares (NNLS) problem, for fixed $\bw $ and $\bu$. In the third step, the value of $\bu$ is updated by solving a non-constrained QP, thereby obtaining the new value of $\bu$ in a closed form in terms of $\bw$ and $\bt$. In the sequel, we present the details of these three steps.

{\bf First step -- updating $\bw$:} We focus on the inner maximization in \eqref{eq:main}, i.e.,
\begin{equation}\label{eq:main2}
\underset{\|\www\| \leq \varepsilon}{\max} \quad \|\GGG\uuu+\www\|^2 + \beta\,\|\HHH(\GGG\uuu+\www) - \Phiii(\ttt)\|^2\, .
\end{equation}
Denoting the maximizer of \eqref{eq:main2} by $\www^*$, it is routine to check, by contradiction, that the norm constraint on $\www$ is active at the optimum, i.e., $\|\www^*\|= \varepsilon$. Thus, the maximization problem \eqref{eq:main2} is equivalent to
\begin{equation}\label{eq:main3}
\underset{\|\www\| = \varepsilon}{\max} \quad \|\GGG\uuu+\www\|^2+\beta\,\|\HHH(\GGG\uuu+\www) - \Phiii(\ttt)\|^2.
\end{equation}
In case $\rank(\HHH)\!>\!1$, no closed-form solution is known for \eqref{eq:main3}. To tackle this problem, we start from its Lagrangian which is given by
\begin{equation}\label{eq:lag}
\begin{aligned}
\mathcal{L}(\www,\tau) =&\; \uuu^\T\GGG^\T\GGG\uuu+\www^\T\www+2\www^\T\GGG\uuu
 + \beta\,(\GGG\uuu+\www)^\T\HHH^\T\HHH(\GGG\uuu+\www) \\
 & + \beta\,\Phiii^\T(\ttt)\Phiii(\ttt)
 - 2\beta\,\Phiii^\T(\ttt)\HHH(\GGG\uuu+\www) - \tau\left(\www^\T\www-\varepsilon^2\right),
\end{aligned}\vspace{-.1cm}
\end{equation}
where $\tau$ is the Lagrange multiplier associated with the norm constraint $\|\www\| = \varepsilon$. Note that since the maximization \eqref{eq:main3} is a non-convex problem, the method of Lagrange multipliers yields only necessary conditions  for optimality which may not be sufficient. Differentiating $\mathcal{L}(\www,\tau)$ with respect to $\www$ and equating it to zero yield
\begin{equation}\label{eq:lagd}
\www^*+\GGG\uuu+\beta\,\HHH^\T\HHH\www^* + \beta\,\HHH^\T\HHH\GGG\uuu - \beta\,\HHH^\T\Phiii(\ttt) - \mu^* \www^* = 0,\vspace{-.2cm}
\end{equation}
and therefore,\vspace{-.1cm}
\begin{equation}\label{eq:eopt}
\www^* = -\left(\PPP - \mu^* \I\right)^{-1} \HHH^\T\big(\GGG\HHH\uuu-\Phiii(\ttt)\big),
\end{equation}
where $\PPP\triangleq\HHH^\T\HHH+\frac{1}{\beta}\I$ and $\mu^*\triangleq\tau^*/\beta$. The maximizer given in \eqref{eq:eopt} must satisfy the norm constraint $\|\www^*\|^2=\varepsilon^2$, i.e.,
\begin{equation}\label{eq:eopt2}
\big(\PPP\GGG\uuu-\HHH^\T\Phiii(\ttt)\big)^\T \left(\HHH^\T\HHH - \mu^* \I\right)^{-2} \big(\PPP\GGG\uuu-\HHH^\T\Phiii(\ttt)\big) = \varepsilon^2,\vspace{-.1cm}
\end{equation}
from which one can obtain $\mu^*$. Let us denote
\begin{equation}\label{eq:func}
f(\mu) \!\triangleq\! \big(\PPP\GGG\uuu-\HHH^\T\Phiii(\ttt)\big)^\T \!\! \left(\HHH^\T\HHH - \mu \I\right)^{-2}\!\! \big(\PPP\GGG\uuu-\HHH^\T\Phiii(\ttt)\big) - \varepsilon^2,
\end{equation}
then $\mu^*$ is a root of $f(\mu)$. Unfortunately, no closed-form solution is known in general for $f(\mu)=0$. Nonetheless, it can be shown that function $f(\mu)$ has a finite number of roots according to the following lemma.
\begin{lemma}\label{lem:1}
        Let $z$ denote the number of roots of $f(\mu)$, then $z$ is always an even number bounded as $$2 \leq z \leq 2\,\rank(\HHH).$$
\end{lemma}
\begin{proof}
        See \cite{slp_relrob}.\end{proof}
Clearly, among all the roots of $f(\mu)$, the one that maximizes the objective function of \eqref{eq:main3} corresponds to the worst-case $\www$, for given $\uuu$ and $\ttt$. The next theorem specifies the interval within which there exists a unique $\mu^*$ yielding the maximizer of \eqref{eq:main3}.

\begin{theorem}\label{thm:1}
        The value of 
         $\mu^*$ is equal the largest positive root of $f(\mu)$ and is bounded as
        \begin{equation}\label{eq:thm22}
        \bar{\lambda}_{\max} < \mu^* \leq \frac{1}{\varepsilon} \left\|\PPP\GGG\uuu-\HHH^\T\Phiii(\ttt)\right\| + \bar{\lambda}_{\max},
        \end{equation}
        with $\bar{\lambda}_{\max} \!\triangleq\!\|\HHH\|^2+\frac{1}{\beta}$.
\end{theorem}
\begin{proof} See 
        \cite{slp_relrob}.
\end{proof}
The above theorem facilitates the possibility of searching for the intended root  of $f(\mu)\!$ in the interval specified by \eqref{eq:thm22} via numerical methods, e.g., a simple bisection search. Using such a numeric solution for $\mu^*$ in \eqref{eq:eopt} yields the optimal value of $\bw$, for given $\bu$ and $\bt \succeq \OOO$, in a semi-closed form   

For rather small values of $\varepsilon$, one can also use quite an accurate approximation for $\mu^*$ with a closed-form expression given below.
\begin{lemma}\label{lem:2}
        For small $\varepsilon$, the value of $\mu^*$ can be well approximated by
        \begin{equation}\label{eq:lem2}
        \mu^*\approx 2\sqrt[3]{ \frac{\|\PPP \left(\PPP\GGG\uuu-\HHH^\T\Phiii(\ttt)\right)\|^2}{\varepsilon^2}}.
        \end{equation}
\end{lemma}
\begin{proof} See
        \cite{slp_relrob}.
\end{proof}
The approximation provided by Lemma \ref{lem:2} is very accurate for $\varepsilon\leq0.1$ based on our observations.
 
{\bf Second step -- updating $\bt$:} For given $\bw$ and $\bu$, the value of $\bt$ is updated as the solution to the following optimization problem: 
\begin{align}\label{eq:main5n}
\underset{\ttt\succeq\OOO}{\min} \quad \left\|\HHH\left(\GGG\uuu+\www\right) - \Phiii(\ttt)\right\|^2,
\end{align}
which is a standard NNLS problem. Note, however, that using the exact solution to \eqref{eq:main5n} in order to update $\ttt$ may result in slow  convergence rate of the iterative method \cite{mm_alg}. One can instead update $\bt$ by using the accelerated projected gradient descent (APGD) algorithm \cite{maj_max}, which provides the update by taking only one step in the steepest descent direction at the current point. 

{\bf Third step -- updating $\bu$:} For given $\bw$ and $\bt \succeq \OOO$, the minimization over $\bu$ is an unconstrained QP and hence is amenable to the following closed-form solution:
\begin{equation}\label{eq:goptn}
\uuu = \GGG^{-1}\PPP^{-1}\HHH^\T \Phiii(\ttt) - \GGG^{-1}\www.
\end{equation}
The pseudo-code of the explained block coordinate ascent-descent algorithm, including the APGD-based updating step of $\ttt$, is provided in Algorithm \ref{alg:1}.

\begin{algorithm}[h]
        \caption{\small{Block coordinate ascent-descent algorithm solving \eqref{eq:main}}}
        \label{alg:1}
        \setstretch{1.75}
        \small{

                \KwInput{$\A, \HHH, \DDD, \s,  \varepsilon$}
                \KwOutput{$\uuu$}
                \KwInitialize{$\ttt^{(0)}=\zzz^{(0)}\in\mathds{R}_+^{2n_\mathrm{r}\times1},\uuu^{(0)}\in\mathds{R}^{2n_\mathrm{t}\times1}, k=0$}
                \KwSet{$\varphi=\frac{1-\sqrt{\kappa}}{1+\sqrt{\kappa}}, \kappa = \frac{\sigma_{\max}}{\sigma_{\min}},\BBB=\I - \sigma^2_{\min} \times\, (\A\A^\T)^{-1}$, \it where $\sigma_\mathrm{max}$ and $\sigma_\mathrm{min}$ respectively denote the maximum and the minimum singular value of matrix $\A$.} 
                \Do{the terminating condition is met}{
                $k = k+1$\\
                \it compute $\mu^{(k)}$ by solving $f(\mu)=0$\\
            $\www^{(k)} = -\left(\PPP - \mu^{(k)} \I\right)^{-1} \left(\PPP\GGG\uuu^{(k-1)}-\DDD\s-\A^{-1}\ttt^{(k-1)}\right)$\\
        $\ttt^{(k)} = \max\left\{\BBB\zzz^{(k-1)} + \sigma^2_{\min} \A^{-\T}\left(\HHH\left(\GGG\uuu^{(k-1)}+\www^{(k)}\right) - \DDD\s \right),\OOO\right\}$\\
        $\zzz^{(k)} = \ttt^{(k)} + \varphi \left(\ttt^{(k)}-\ttt^{(k-1)}\right)$\\
        $\uuu^{(k)} = \GGG^{-1}\PPP^{-1}\HHH^\T \left(\DDD\s+\A^{-1}\ttt^{(k)}\right) - \GGG^{-1}\www^{(k)}$\\
 }   
}
\setstretch{1}
\nonl{\footnotesize{}}

\end{algorithm}

To provide an intuition of the structure of the optimal transmit signal, let $(\www^*,\uuu^*,\ttt^*)$ denote the solution to  \eqref{eq:main}. It then follows from \eqref{eq:goptn} that
\begin{equation}\label{eq:sigr}
\GGG\uuu^* + \www^* = \left(\HHH^\T\HHH+\frac{1}{\beta}\I\right)^{-1}\HHH^\T\left(\DDD\s+\A^{-1}\ttt^*\right),
\end{equation}
i.e., the optimal worst-case robust transmit signal can simply be viewed as applying a (regularized) channel inversion to the constructively-interfered symbols, with the interference components being aligned such that the received symbols are pushed (as deep as possible) into the CI regions. Furthermore, considering the limiting case $\beta\rightarrow\infty$, in which $\PPP^{-1}\HHH^\T=\HHH^\dagger$, implies that for extremely large values of $\beta$, the received symbol of each UE is guaranteed to be observed within the correct CI region, even for the worst possible error realization. Note, however, that this limiting case $\beta$ may cause an unaffordable transmission power.



\section{Simulation Results}\label{sec:sim}

 The simulation setup is as follows. We consider a downlink MU-MIMO system with $n_\mathrm{t}\!=\!n_\mathrm{r}\!=\!8$, in which unit noise variances $\sigma_i^2\!=\!1$ and equal target SNRs $\gamma_i\!\triangleq\!\gamma$ are assumed for all $i=1,...,n_\mathrm{r}$. Assuming a Rayleigh block fading channel, the channel vectors $\{\h_i\}_{i=1}^{n_\mathrm{r}}$ are independently generated for each coherence block following the standard circularly symmetric complex Gaussian (CSCG) distribution, i.e., $\h_i\!\sim\!\mathcal{CN}(\OOO,\I)$. All our simulation results are averaged over $500$ channel coherence blocks each with $500$ symbols.  We refer to our proposed worst-case SLP design as WC-SLP.

The additive distortion vector $\www$ is randomly generated as an i.i.d. CSCG vector with variance $0.1$. The distortion ball radius is set to be $\varepsilon=0.56$, which corresponds to a confidence level of $0.99$, i.e., $\PR\{\|\www\|>\varepsilon\}=0.01$. We further assume $\GGG=\I$. In our simulations, we have defined energy efficiency as the ratio of the product of the average UEs' bit error rate (BER) and the per-user achievable rate divided by the total consumed power (i.e., $\uuu^\T\uuu$). The achievable rate $I(s_i;r_i)$ for the $i$th UE can be obtained as
\begin{align}\label{eq:mi}
I(s_i;r_i) = \EXP_{s_i,r_i,\HHH}\left\{\log_2 \frac{P_{r_i|s_i,\HHH}(r_i|s_i,\HHH)}{P_{r_i|\HHH}(r_i|\HHH)}\right\}.
\end{align}
The conditional probability mass functions in \eqref{eq:mi} are not amenable to closed-form expressions. To tackle this difficulty, inspired by \cite{sven_thro}, we resort to empirical probability distributions obtained by generating sufficiently many channel and symbol realizations, and then compute an approximation (in fact, a lower bound) for the mutual information in \eqref{eq:mi}.

The energy efficiency of the WC-SLP scheme is plotted in Fig. \ref{fig:all} as a function of the UEs' target SNR, for different values of $\beta$. To have a benchmark for comparison, the results for the SLP problem \eqref{eq:slp} in the absence of any distortion are also presented. Among all the values of $\beta$ shown in Fig. \ref{fig:all}, choosing $\beta=1$ results in a higher energy efficiency, even compared to the ideal undistorted SLP. \textit{This is a consequence of relaxing the CI constraints in the SLP problem, leading to a lower transmit power in exchange for a slightly higher BER}. Increasing $\beta$, on the other hand, reduces the energy efficiency of the proposed WC-SLP scheme. This can simply be justified by considering the limiting case $\beta\rightarrow\infty$, in which the design formulation \eqref{eq:slp2} aims to strictly impose the CI constraints, regardless of the required transmit power. In general, a proper choice of $\beta$ is application-dependent and relies on the corresponding system/user requirements. For instance, in wireless systems with strict target BERs, a larger $\beta$ is more preferred. On the other hand, in scenarios where transmit power is strictly limited, one may choose smaller values for $\beta$. Moreover, the value of $\beta$ can be adjusted in a more sophisticated way, e.g., letting $\beta$ vary as a function of the target SNR $\gamma$, which is the topic of an ongoing research.

\begin{figure}[t]
	\centering
	\includegraphics[width=.5\columnwidth]{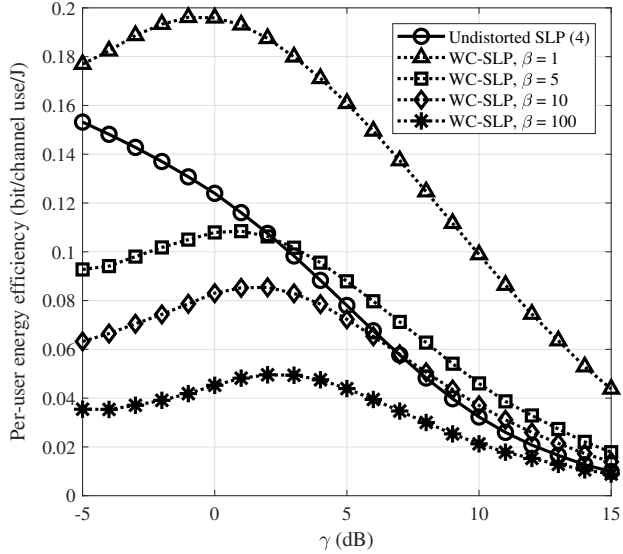}
	\caption{Performance comparison of distorted and undistorted SLP schemes as a function of target SNR.}
	\label{fig:all}
\end{figure}

\section{Conclusions}\label{sec:con}

In this paper, we proposed a worst-case design formulation for the QoS-constrained SLP problem minimizing the total transmit power in a scenario where the precoder's output \textit{undergoes linear distortion with bounded additive noise.} A new problem formulation was first proposed, which led us to cast the worst-case design of the distorted SLP as a min-max problem by introducing \textit{relaxed CI constraints}. We then solved this problem using an iterative  coordinate ascent-descent algorithm to obtain the robust precoded signal. This algorithm iterates between finding the optimal precoded signal and the worst-case additive distortion vector. Finding the precoded signal involves solving a non-negative least squares problem, while obtaining the worst-case distortion vector leads to a semi-closed form solution with only one  scalar parameter which has to be calculated numerically. Our simulation results show that the proposed worst-case approach can outperform the undistorted SLP
method
in terms of energy efficiency.

\section{Acknowledgment}

The authors would like to thank Dr.~Farbod Kayhan for facilitating this work and for his suggestions and ideas which helped us to improve this paper.


\begin{thebibliography}{10}
	\providecommand{\url}[1]{#1}
	\csname url@samestyle\endcsname
	\providecommand{\newblock}{\relax}
	\providecommand{\bibinfo}[2]{#2}
	\providecommand{\BIBentrySTDinterwordspacing}{\spaceskip=0pt\relax}
	\providecommand{\BIBentryALTinterwordstretchfactor}{4}
	\providecommand{\BIBentryALTinterwordspacing}{\spaceskip=\fontdimen2\font plus
		\BIBentryALTinterwordstretchfactor\fontdimen3\font minus
		\fontdimen4\font\relax}
	\providecommand{\BIBforeignlanguage}[2]{{%
			\expandafter\ifx\csname l@#1\endcsname\relax
			\typeout{** WARNING: IEEEtran.bst: No hyphenation pattern has been}%
			\typeout{** loaded for the language `#1'. Using the pattern for}%
			\typeout{** the default language instead.}%
			\else
			\language=\csname l@#1\endcsname
			\fi
			#2}}
	\providecommand{\BIBdecl}{\relax}
	\BIBdecl
	
	\bibitem{tb_opt}
	M.~Bengtsson and B.~Ottersten, \emph{Handbook of Antennas in Wireless
		Communications}, 2001, ch. Optimal and suboptimal transmit beamforming.
	
	\bibitem{vec_per}
	C.~B. Peel, B.~M. Hochwald, and A.~L. Swindlehurst, ``A vector-perturbation
	technique for near-capacity multiantenna multiuser communication-part
	\uppercase{I}: channel inversion and regularization,'' \emph{IEEE Trans.
		Commun.}, vol.~53, no.~1, pp. 195--202, Jan. 2005.
	
	\bibitem{tb_convex}
	A.~B. Gershman, N.~D. Sidiropoulos, S.~Shahbazpanahi, M.~Bengtsson, and
	B.~Ottersten, ``Convex optimization-based beamforming,'' \emph{IEEE Signal
		Process. Mag.}, vol.~27, no.~3, pp. 62--75, May 2010.
	
	\bibitem{tb_sinr}
	M.~Schubert and H.~Boche, ``Solution of the multiuser downlink beamforming
	problem with individual \uppercase{SINR} constraints,'' \emph{IEEE Trans.
		Veh. Technol.}, vol.~53, no.~1, pp. 18--28, Jan. 2004.
	
	\bibitem{slp_chr}
	C.~Masouros and G.~Zheng, ``Exploiting known interference as green signal power
	for downlink beamforming optimization,'' \emph{IEEE Trans. Signal Process.},
	vol.~63, no.~14, pp. 3628--3640, Jul. 2015.
	
	\bibitem{slp_con}
	M.~Alodeh, S.~Chatzinotas, and B.~Ottersten, ``Constructive multiuser
	interference in symbol level precoding for the \uppercase{MISO} downlink
	channel,'' \emph{IEEE Trans. Signal Process.}, vol.~63, no.~9, pp.
	2239--2252, May 2015.
	
	\bibitem{slp_gen}
	A.~Haqiqatnejad, F.~Kayhan, and B.~Ottersten, ``Constructive interference for
	generic constellations,'' \emph{IEEE Signal Process. Lett.}, vol.~25, no.~4,
	pp. 586--590, Apr. 2018.
	
	\bibitem{sp_over}
	R.~W. {Heath}, N.~{González-Prelcic}, S.~{Rangan}, W.~{Roh}, and A.~M.
	{Sayeed}, ``An overview of signal processing techniques for millimeter wave
	\uppercase{MIMO} systems,'' \emph{IEEE J. Sel. Topics in Signal Process.},
	vol.~10, no.~3, pp. 436--453, Apr. 2016.
	
	\bibitem{unc_chr}
	K.~L. Law and C.~Masouros, ``Constructive interference exploitation for
	downlink beamforming based on noise robustness and outage probability,'' in
	\emph{2016 IEEE Int. Conf. Acoustics, Speech and Signal Process. (ICASSP)},
	Mar. 2016, pp. 3291--3295.
	
	\bibitem{slp_robust}
	A.~{Haqiqatnejad}, F.~{Kayhan}, and B.~{Ottersten}, ``Robust design of power
	minimizing symbol-level precoder under channel uncertainty,'' in \emph{2018
		IEEE Global Commun. Conf. (GLOBECOM)}, Dec. 2018, pp. 1--6.
	
	\bibitem{mehrzad2004rdp}
	M.~Biguesh, S.~Shahbazpanahi, and A.~B. Gershman, ``{Robust downlink power
		control in wireless cellular systems},'' \emph{EURASIP Journal on Wireless
		Communications and Networking}, vol. 2004, pp. 261--272, Dec. 2004.
	
	\bibitem{Chalise2007}
	B.~K. Chalise, S.~Shahbazpanahi, A.~Czylwik, and A.~B. Gershman, ``Robust
	downlink beamforming based on outage probability specifications,''
	\emph{{IEEE} Trans. Wireless Commun.}, vol.~6, no.~10, pp. 3498--3503, Oct.
	2007.
	
	\bibitem{Shahbazpanahi2003}
	S.~Shahbazpanahi, A.~B. Gershman, Z.-Q. Luo, and K.~M. Wong, ``Robust adaptive
	beamforming for general-rank signal models,'' \emph{{IEEE} Trans. Signal
		Process.}, vol.~51, no.~9, pp. 2257--2269, Sep. 2003.
	
	\bibitem{Vorobyov2}
	S.~Vorobyov, ``Robust CDMA multiuser detectors: Probability-constrained versus
	the worst-case-based design,'' \emph{{IEEE} Signal Process. Lett.}, vol.~15,
	pp. 273 --276, Nov. 2008.
	
	\bibitem{Vorobyov4}
	S.~Vorobyov, H.~Chen, and A.~Gershman, ``On the relationship between robust
	minimum variance beamformers with probabilistic and worst-case distortionless
	response constraints,'' \emph{{IEEE} Trans. Signal Process.}, vol.~56, pp.
	5719 --5724, Nov. 2008.
	
	\bibitem{ShahramMUD}
	S.~{Shahbazpanahi} and A.~B. {Gershman}, ``Robust blind multiuser detection for
	synchronous CDMA systems using worst-case performance optimization,''
	\emph{{IEEE} Trans. Wireless Commun.}, vol.~3, no.~6, pp. 2232--2245, Nov
	2004.
	
	\bibitem{Zarifi}
	K.~{Zarifi}, S.~{Shahbazpanahi}, A.~B. {Gershman}, and {Zhi-Quan Luo}, ``Robust
	blind multiuser detection based on the worst-case performance optimization of
	the MMSE receiver,'' \emph{{IEEE} Trans. Signal Process.}, vol.~53, no.~1,
	pp. 295--305, Jan 2005.
	
	\bibitem{YueShahramOSTBC}
	{Yue Rong}, S.~{Shahbazpanahi}, and A.~B. {Gershman}, ``Robust linear receivers
	for space-time block coded multiaccess MIMO systems with imperfect channel
	state information,'' \emph{{IEEE} Trans. Signal Process.}, vol.~53, no.~8,
	pp. 3081--3090, Aug 2005.
	
	\bibitem{slp_tsp}
	A.~Haqiqatnejad, F.~Kayhan, and B.~Ottersten, ``Symbol-level precoding design
	based on distance preserving constructive interference regions,'' \emph{IEEE
		Trans. Signal Process.}, vol.~66, no.~22, pp. 5817--5832, Nov. 2018.
	
	\bibitem{sat_hybrid}
	C.~M. Jes{\'u}s~Arnau, Bertrand~Devillers and A.~P{\'e}rez-Neira, ``Performance
	study of multiuser interference mitigation schemes for hybrid broadband
	multibeam satellite architectures,'' \emph{EURASIP J. Wirel. Commun. Netw.},
	vol. 2012, no.~1, p. 132, Apr. 2012.
	
	\bibitem{vec_quan}
	A.~Gersho and R.~M. Gray, \emph{Vector quantization and signal
		compression}.\hskip 1em plus 0.5em minus 0.4em\relax Springer Science \&
	Business Media, 2012, vol. 159.
	
	\bibitem{pre_quan}
	A.~K. {Fletcher}, S.~{Rangan}, V.~K. {Goyal}, and K.~{Ramchandran}, ``Robust
	predictive quantization: Analysis and design via convex optimization,''
	\emph{IEEE J. Sel. Topics Signal Process.}, vol.~1, no.~4, pp. 618--632, Dec
	2007.
	
	\bibitem{slp_relrob}
	A.~Haqiqatnejad, S.~Shahbazpanahi, F.~Kayhan, and B.~Ottersten, ``Robust
	symbol-level precoding for downlink MU-MIMO under design uncertainty: A
	worst-case performance optimization based approach,'' \emph{Journal paper},
	In preparation.
	
	\bibitem{mm_alg}
	D.~R. Hunter and K.~Lange, ``A tutorial on MM algorithms,'' \emph{The American
		Statistician}, vol.~58, no.~1, pp. 30--37, 2004.
	
	\bibitem{maj_max}
	Y.~{Sun}, P.~{Babu}, and D.~P. {Palomar}, ``Majorization-minimization
	algorithms in signal processing, communications, and machine learning,''
	\emph{IEEE Trans. Signal Process.}, vol.~65, no.~3, pp. 794--816, Feb. 2017.
	
	\bibitem{sven_thro}
	S.~{Jacobsson}, G.~{Durisi}, M.~{Coldrey}, U.~{Gustavsson}, and C.~{Studer},
	``Throughput analysis of massive \uppercase{MIMO} uplink with low-resolution
	\uppercase{ADC}s,'' \emph{IEEE Trans. Wirel. Commun.}, vol.~16, no.~6, pp.
	4038--4051, Jun. 2017.
	
\end{thebibliography}
\end{document}